\documentclass[a4paper,twocolumn,fontsize=10pt,DIV=16,abstract=true,headings=normal]{scrartcl}
\pdfoutput=1

\RedeclareSectionCommand[
  afterindent=false,
  beforeskip=.5\baselineskip,
  afterskip=.25\baselineskip]{section}
\RedeclareSectionCommand[
  afterindent=false,
  beforeskip=.25\baselineskip,
  afterskip=.125\baselineskip]{subsection}
\RedeclareSectionCommand[
  afterindent=false,
  beforeskip=.5\baselineskip,
  afterskip=.125\baselineskip]{subsubsection}

\usepackage[utf8]{inputenc}
\usepackage[T1]{fontenc}
\usepackage[protrusion=true,expansion=true]{microtype}

\usepackage[backend=biber,alldates=edtf,seconds=true]{biblatex}
\usepackage{balance}

\usepackage[acronym]{glossaries}

\newcommand{\nc}{\newcommand}
\nc{\PP}{\mathbb{P}}
\nc{\BE}{\mathbb{E}}
\nc{\BV}{\mathbb{V}}

\usepackage{tikz}
\usetikzlibrary{arrows.meta,shapes}

\RequirePackage{booktabs}
\RequirePackage{longtable}
\RequirePackage{array}
\RequirePackage{tabularx}
\RequirePackage{multicol}
\RequirePackage{multirow}

\RequirePackage{siunitx}
\usepackage{scrlayer-scrpage}

\usepackage{amsmath}
\usepackage{amssymb}
\usepackage{amsfonts}
\usepackage{amsthm}
\DeclareMathSizes{10}{9}{7}{5}

\usepackage{listings}
\usepackage[capitalise]{cleveref}
\crefname{subsection}{Subsection}{Subsections}

\usepackage{xspace}

\usepackage{enumitem}
\setitemize{itemsep=-1mm,topsep=0mm,leftmargin=*}
\usepackage{caption}
\usepackage{color}

\lstdefinelanguage{crdt-code}{
  keywords={state, pre, let, update, query, effector, generator, initial, return, if, else, for},
  keywordstyle=\bfseries,
  sensitive=false,
  comment=[l]{//},
  morecomment=[s]{/*}{*/},
  commentstyle=\color{gray}\ttfamily,
  stringstyle=\color{red}\ttfamily,
  morestring=[b]',
  morestring=[b]"
}

\newtheorem{theorem}{Theorem}
\newtheorem{lemma}{Lemma}

\newacronym{DLT}{DLT}{Distributed Ledger Technology}
\newacronym{MEG}{MEG}{Matrix Event Graph}
\newacronym{CRDT}{CRDT}{Conflict-Free Replicated Data Type}
\newacronym{DAG}{DAG}{Directed, Acyclic Graph}
\newacronym{SEC}{SEC}{Strong Eventual Consistency}

\begin{document}
\setlength{\abovedisplayskip}{2mm}
\setlength{\belowdisplayskip}{2mm}
\setlength{\abovedisplayshortskip}{2mm}
\setlength{\belowdisplayshortskip}{2mm}

\title{Analysis of the Matrix Event Graph\texorpdfstring{\\ }{} Replicated Data Type}

\author{
Florian Jacob\\
\footnotesize Karlsruhe Institute\\[-3mm]
\footnotesize of Technology\\[-3mm]
\footnotesize Institute of Telematics\\[-3mm]
\footnotesize florian.jacob@kit.edu
\and
Carolin Beer\\
\footnotesize Karlsruhe Institute\\[-3mm]
\footnotesize of Technology\\[-3mm]
\footnotesize Institute of Telematics\\[-3mm]
\footnotesize carolin.beer@student.kit.edu
\and
Norbert Henze\\
\footnotesize Karlsruhe Institute\\[-3mm]
\footnotesize of Technology\\[-3mm]
\footnotesize Institute of Stochastics\\[-3mm]
\footnotesize henze@kit.edu
\and
Hannes Hartenstein\\
\footnotesize Karlsruhe Institute\\[-3mm]
\footnotesize of Technology\\[-3mm]
\footnotesize Institute of Telematics\\[-3mm]
\footnotesize hannes.hartenstein@kit.edu
}

\maketitle

\begin{abstract}
Matrix is a new kind of decentralized, topic-based publish-subscribe middleware for communication and data storage that is getting popular particularly as a basis for secure instant messaging.
In comparison to traditional decentralized communication systems,
Matrix replaces pure message passing with a replicated data structure.
This data structure, which we extract and call the Matrix Event Graph (MEG),
depicts the causal history of messages.
We show that this MEG represents an interesting and important replicated data type for general decentralized applications that are based on causal histories of publish-subscribe events: we show that a MEG possesses strong properties with respect to consistency, byzantine attackers, and scalability.
First, we show that the MEG provides Strong Eventual Consistency (SEC), and that it is available under partition,
by proving that the MEG is a Conflict-Free Replicated Data Type for causal histories.
While strong consistency is impossible here
as shown by the famous CAP theorem, SEC is among the best known achievable trade-offs.
Second, we discuss the implications of byzantine attackers on the data type's properties.
We note that the MEG, as it does not strive for consensus,
can cope with $n > f$ environments with $n$ total participants of which $f$ show byzantine faults.
Furthermore, we analyze scalability: Using Markov chains
we study the width of the MEG, defined as the number of forward extremities, over time and observe an almost optimal evolution.
We conjecture that this property is inherent to the underlying spatially inhomogeneous random walk.
\end{abstract}

\section{Introduction}
Matrix\footnote{\url{https://matrix.org/}, \url{https:/matrix.org/spec/}} is a specification
of protocols and their behavior for a middleware that provides communication and data services for decentralized applications.
While the size of its public federation is still comparatively small,
its utilization rises quickly,
and several organizations are deploying large, private federations.
Currently, Matrix is mainly used as the basis of a decentralized instant messaging protocol
employed by the French government, the Mozilla foundation, the Federal Defense Forces of Germany, and others.

Matrix implements topic-based publish-subscribe services based on a federated architecture.
Similar to e-mail or XMPP, clients attach themselves to a Matrix server,
their so-called homeserver, which represents them in the Matrix network.
Servers with clients subscribed to a specific topic (called room in Matrix parlance) form a federation to exchange published events independent of other topics.
Events can be either communication events or state update events on the stored data.
In the instant messaging use case, topics are employed for group or one-to-one communication rooms,
communication events are used for instant messages, while the stored data is used for persistent information like room membership or room description.

In contrast to e-mail or XMPP, Matrix replaces pure message passing
with a replicated, per-topic data structure that stores the causal history of events.
As Matrix servers can thereby synchronize their room's full causal histories,
the Matrix approach promises increased decentralized system resilience:
After a network partition, a server has significantly stronger means to recover the complete state of the room, i.e., to avoid loss of events.
While this increased level of system resilience has been observed by practitioners, the underlying replicated data type has not yet been analyzed thoroughly.

\glsunset{MEG}
In this paper, we first extract and abstract the \acrlong{MEG} replicated data type from the Matrix
specification and denote it by \gls{MEG}.
A \gls{MEG} is a \gls{DAG} made up of vertices
which represent communication and data storage update events,
and directed edges which stand for potential causal relations between events.

Because the graph represents the potential causal order of events,
a correct graph is inherently cycle-free.
Appending new events is the only write operation supported by the Matrix Event Graph,
which makes it append-only --- and a candidate for Distributed Ledger Technologies. Thus, the \gls{MEG} can be considered as a fundamental concept for various applications that are based on causal histories, ranging from decentralized crowdsensing databases in Internet of Things scenarios over decentralized collaboration applications to decentralized push notification systems.
Since, for Distributed Ledger Technologies, it has been conjectured that consistency,
decentralization, and scalability cannot be achieved simultaneously~\cite{Zhang2018, Raikwar2020}, our analysis focuses on these aspects.

As main contribution we therefore provide an analysis of the degree to which the \gls{MEG} fulfills consistency, deployability in decentralized scenarios, and scalability:

{\bfseries Consistency:} In accordance with the CAP theorem~\cite{gilbert-lynch-cap-proof}, and since
Matrix provides availability and partition tolerance,
the \gls{MEG} necessarily has to sacrifice strong consistency.
We show that Matrix provides Strong Eventual Consistency
by proving that the \gls{MEG} is a \gls{CRDT}~\cite{shapiro-crdt} for causal histories.

{\bfseries Decentralization:}
We discuss the implications of byzantine attackers on the specific type of \gls{CRDT} that the \gls{MEG} represents.
The avoidance of consensus is the primary reason that
allows the \gls{MEG} \gls{CRDT} to facilitate $n > f$ environments with $n$ total participants of which $f$ exhibit byzantine faults.

{\bfseries Scalability:}
The inherent probabilism of uncoordinated, concurrent updates on a \gls{MEG}
is the main challeng for the analysis of the \gls{MEG} with respect to scalability.
We are interested in the width of the \gls{MEG} in terms of the number of forward extremities, i.e. `vertices without children', over time.
We study the width of the \gls{MEG} using a formalization by means of Markov chains.
We observe that the \gls{MEG} does not degenerate,
and conjecture that this non-degeneracy is inherent to the underlying spatially inhomogeneous random walk.

This paper is structured as follows:
We start with a more detailed description of how the \gls{MEG} works and the problem statement in \cref{sec:mx_event_graph}.
\cref{sec:background} presents related work and background on replicated data types.
Assumptions and architecture are given in \cref{sec:system-model}.
The inner working of the \gls{MEG} is formalized in \cref{sec:mx-crdt},
which is then used to prove that it is a Conflict-Free Replicated Data Type.
In \cref{sec:weakening-assumptions},
we perform a reality check of the utilized assumptions of \cref{sec:mx-crdt} and discuss how the \gls{MEG} can be made byzantine fault tolerant.
\cref{sec:convergence} formalizes the stochastic behavior of the width of the \gls{MEG}
and provides evidence that the width always evolves to a near-optimal value, and does so fast.
We conclude the paper in \cref{sec:conclusion}.

\begin{figure}[htbp]
  \resizebox{0.85\linewidth}{!}{
    \includegraphics{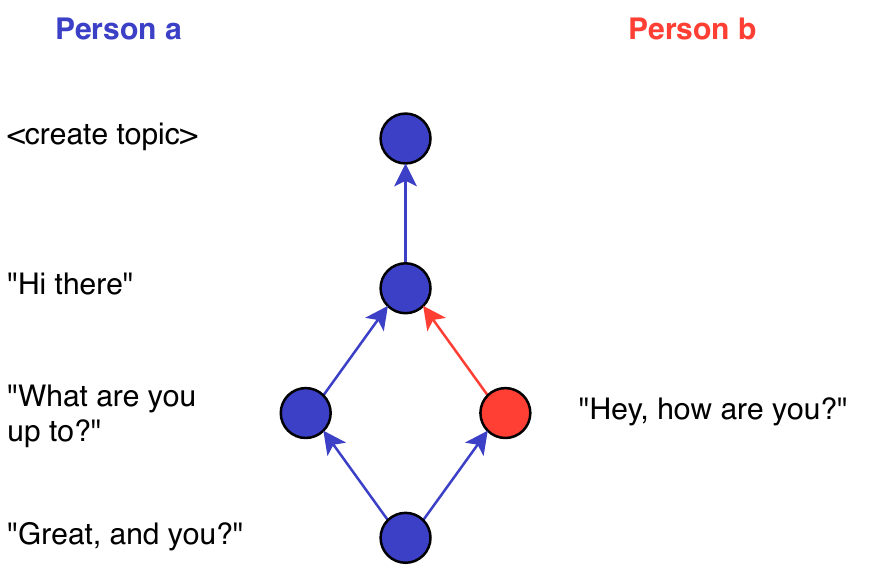}
  }
  \caption{Basic example of a \gls{MEG}}
  \label{fig:event_graph}
\end{figure}

\section{MEG: Overview and Problem Statement}
\label{sec:mx_event_graph}

In the following and for illustration purposes, we often make use of the instant messaging use case of Matrix,
but we want to emphasize that the \acrlong{MEG} is a general replicated data type for append-only causal histories of publish-subscribe events.
We also typically focus our studies on a single \gls{MEG} instance, and therefore a single broadcast domain associated with that \gls{MEG}.
However, several independent \glspl{MEG} can coexist.
A sample \gls{MEG} is exhibited in \cref{fig:event_graph}.

{\bfseries General \gls{MEG} setup.} As mentioned before, a \gls{MEG} is a \acrfull{DAG}.
One \gls{MEG} represents the message history and attributes of a group or 1:1 chat,
and it is replicated independently by all participating servers.
Upon creation, the \gls{DAG} consists of only a single vertex, the \emph{root vertex}.
Each vertex in the \gls{DAG} corresponds to an application-defined publish-subscribe event,
e.g., to a text message or temperature reading.
Edges represent potential causal relationships between events:
When a new vertex is added,
it is appended to the existing \gls{DAG} through one or more outgoing edges.
These edges point towards vertices that had no incoming edges before,
i.e., the newest events in causal history,
which we from now on call the \emph{forward extremities} of the DAG.
The selection of forward extremities is done according to the current knowledge of the adding replica.
This potential causal relationship is known as the \emph{happened before} relationship\footnote{Note that \citeauthor{Lamport1978} defines $a$ \emph{happened before} $b$ as $a \rightarrow b$.
In this paper, we actually use the converse relation $b \rightarrow a$, as common for Distributed Ledger Technologies,
so that new references can be stored as part of the new vertex,
and old vertices can be kept immutable.
It follows that for $b \rightarrow a$, we say $a$ is the parent of $b$.
}, as defined by Lamport~\cite{Lamport1978}:
For $a \leftarrow b$, we say $a$ happened before $b$.
Edges thereby form a partial order that is consistent with the causal order in which events took place.

In addition to being directed, acyclic, and representing the causal order of events, the \gls{MEG}  is also weakly connected since all newly added vertices have at least one outgoing edge.
The  root vertex, as the only vertex without outgoing edges,
is therefore the unique minimal element of the partial order represented by the \gls{DAG}.
DAGs with this specific structure are called \emph{rooted}~\cite{Platt2000}.

{\bfseries Adding a new vertex to the source replica.}
The replica that creates an event on behalf of a client and appends it as a vertex is called \emph{source replica}.
When it adds a vertex,
the corresponding event could be causally related to previous events.
Thus, all forward extremities should be included as edges.
Replicas can experience a high number of forward extremities caused by latencies or partitions,
and malicious replicas could forge events with a high number of parents.
However, certain algorithms executed on the \gls{MEG} do not scale well with the number of parent events,
i.e., they can become very resource intensive, especially when old parts of the \gls{MEG} are referenced as parents~\cite{synapse-issue-forward-extremities-accumulate}.
In practice, the maximum number of parent events
therefore has to be restricted to a finite value $d$.
If there are more than $d$ forward extremities,
a replica selects a subset of size $d$ for the new event.
For the potential causal order relation in the \gls{MEG} still to be consistent with the actual causal order,
clients have to inform the replica about actual causal dependencies so that those are included as parents.

{\bfseries Updating all replicas.} Beyond appending the new vertex to the local \gls{DAG}, the source replica also needs to synchronize with the other replicas.
The replica sends a \gls{DAG} update that consists of the new vertex and edges
to all replicas using a broadcast protocol.
On reception of an update,
replicas append the new vertex to their \gls{DAG} via the new edges selected by the source replica as soon as all required parent vertices exist in the local replica.
In case the parent vertices are not (yet) available, the update is buffered until they are.

{\bfseries Dealing with concurrent updates.} When clients at two different replicas concurrently invoke updates, each replica thinks of their vertex as the single next step in causal history represented by their \glspl{DAG},
i.e., both deviate from the last consistent \gls{DAG} state.
In case of continuous synchronization failure,
e.g. due to a network partition,
additional client updates will enlarge the inconsistency between the replicas' DAGs
and lead to two causally independent chains of events, built from the last synchronized event.
Both replicas will continue to try to synchronize their state with other replicas.
When the partition heals, all replicas will eventually receive all updates.
As depicted in \cref{fig:event_graph}, instead of trying to find a linear order of updates and to solve conflicts with rollbacks,
the concurrent DAG states are merged by attaching both causally independent chains of events to the last synchronized event, i.e., by forking the DAG.
This acceptance of concurrency in the data type itself by only providing a partial order on events is the core idea of the Matrix Event Graph.
It is also the basis for our proof of conflict-freedom in \cref{sec:mx-crdt}.
A fork in the DAG introduced by concurrency will lead to two causally independent forward extremities.
Following the attachment rules for new vertices,
a replica that has received and appended both causally independent chains to its DAG selects both as parents for a new vertex.
In terms of graphs, this means that the new vertex will join both chains again, which marks that the period of concurrency and causal independence is over,
and reduces the number of forward extremities by one.

{\bfseries Problem statement.}
The way in which concurrency is handled in a \gls{MEG} as well as the use of various parameters as outlined above give rise to the key research questions addressed in this paper:
Which consistency guarantees can application developers expect from a \gls{MEG} --- and under which assumptions do they hold? And: Can the width of the \gls{MEG} degenerate?
The preceding explanations describe how the \gls{MEG} is available under partition,
and how it tries to achieve Eventual Consistency,
as conjectured by the Matrix developers~\cite{matrix-spec-architecture}.
In this paper, we provide a proof of Strong Eventual Consistency in \cref{sec:mx-crdt}.
In \cref{sec:weakening-assumptions}, we relax the employed assumptions, particularly on the communication primitive.
In addition, the overview above
showed that if the number of vertex parents is restricted to $d$ and selected randomly,
the evolution of the number of forward extremities $u$,
i.e., the width of the DAG, is non-trivial in concurrent environments.
In \cref{sec:convergence},
we explore whether
for arbitrary start values of $u$,
if $k$ replicas continuously select $d$ parents independently and then synchronize the new vertices,
the width of the DAG converges in a sufficiently small number of iterations.
In addition, we explore how the choice of the number of parent vertices $d$ affects the speed of convergence.

{\bfseries Not in scope of this paper:}
While we make assumptions on and deal with the underlying broadcast communication primitive,
we consider the topic of broadcast communication per se beyond the scope of this paper.
Moreover, Matrix employs an access control system for \glspl{MEG}, which we will not consider further, but which has been examined in \cite{matrix-decomposition}.

\section{Related Work \& Background}
\label{sec:background}
\glsreset{CRDT}

\citeauthor*{glimpseofthematrix} investigated quantitative aspects of the public Matrix federation,
and found scalability problems with the broadcast communication
currently employed by Matrix~\cite{glimpseofthematrix}.
However, they did not investigate the scalability and other properties of the replicated data structure itself.
The access control system of Matrix, which builds on top of the \gls{MEG}, was very recently studied in \cite{matrix-decomposition}.
Privacy and usability aspects of Matrix,
along with a \acrshort{CRDT}-based vision on how to improve this situation in federated networks in general,
are the topic of~\cite{auvolat2019}.

In the field of replicated data types,
\citeauthor*{shapiro-crdt} introduced the category of \glspl{CRDT},
together with a new consistency model provided by the category, namely Strong Eventual Consistency~\cite{shapiro-crdt}.
Following the initial definition,
new papers mostly focused on implementations of the data type like the JSON-CRDT by \citeauthor*{kleppmann-json-crdt} ~\cite{kleppmann-json-crdt},
or extended the base concept of \glspl{CRDT}~\cite{DePorre2019}.

The initial \gls{CRDT} concept was overhauled in cooperation with the original authors in~\cite{Preguica2019}.
We will mainly use the new \gls{CRDT} terminology introduced there.

\subsection{Consistency Models}
\label{subsec:consistency}
The inherent trade-off between \emph{Consistency} and \emph{Availability} in the presence of network partitions in distributed systems led to the definition of a variety of consistency models.
A well-known consistency model is \emph{Eventual Consistency} (EC), which provides the following guarantees~\cite{shapiro-crdt}:
\begin{itemize}
	\item \emph{Eventual Delivery:}
		An update applied by one correct replica is eventually applied by every correct replica.
	\item \emph{Termination:} Every invoked method terminates.
	\item \emph{Convergence:} Correct replicas that applied the same set of updates eventually reach equivalent states.
\end{itemize}

\emph{\gls{SEC}} builds on top of EC,
and strengthens Convergence~\cite{shapiro-crdt}:

\begin{itemize}
	\item \emph{Strong Convergence:} Correct replicas that applied the same set of updates have equivalent states.
\end{itemize}

Whether two states are equivalent is application-dependent.
In our case, the state of two replicas is equivalent
if their graphs consist of identical vertices and edges.
Note that “the same \emph{set} of updates” means that while the updates are identical, they might be received or applied in different order.
The key difference between Convergence and Strong Convergence is that with Convergence,
replicas may coordinate with other replicas to find agreement on their state even after having applied updates.
Especially if the ordering of updates matters, this can lead to rollbacks.
With Strong Convergence, the agreement has to be immanent and implicit.

\subsection{Conflict-Free Replicated Data Types}
\label{subsec:crdt}
\glsreset{CRDT}
\glspl{CRDT} were first formalized in~\cite{shapiro-crdt}.
\glspl{CRDT} are an abstract data structure that allows for optimistic update execution (cf.~\cite{Saito2005})
while guaranteeing conflict-freedom upon network synchronization.
The system model of \glspl{CRDT} is based on a \emph{fail-silent} abstraction with a Causal Order Reliable Broadcast communication protocol (see \cref{sec:system-model}).
For objects that implement a \gls{CRDT} in a system with $n$ replicas,
\citeauthor*{shapiro-crdt} show that \gls{SEC} is ensured for up to $n-1$ replica failures~\cite{shapiro-crdt}.

Two conceptually different, but equally expressive types of \glspl{CRDT} are the \emph{operation-based} and the \emph{state-based} \gls{CRDT}.
Replicas implement functions to be invoked by clients to access or modify the state.
The key difference between operation- and state-based \glspl{CRDT} lies in the way of synchronization:
In state-based \glspl{CRDT}, all replicas periodically send their full state to all other replicas which then merge states.
In contrast, operation-based \glspl{CRDT} only synchronize upon changes.
Source replicas transmit state changes resulting from a client invocation as operations.
In \cref{sec:mx-crdt}, we show that the \gls{MEG} is an operation-based \gls{CRDT}.

Operation-based \glspl{CRDT} implement functions that can be classified as \texttt{update} or \texttt{query}.
A \texttt{query} function returns information on the current state of the replica.
Their counterpart, \texttt{update} functions, modify the state.
They comprise two steps:
At first, a \texttt{generator}\footnote{Originally introduced as \emph{prepare-update}} step is executed by the source replica.
It is side-effect-free, but returns an \emph{operation},
i.\,e., an encapsulation of the state changes.
A common example of a \texttt{generator} step is the creation of a unique object identifier for \texttt{update} functions that add an object to the state.
The second step is called \texttt{effector}\footnote{Originally introduced as \emph{effect-update}} step, it must be executed at every replica.
Thus, the source replica transmits the generated operation to all replicas using broadcast.
Upon reception of an operation,
each replica executes the \texttt{effector} step locally and applies the resulting changes to their state.~\cite{Preguica2011}


In general, the data structure of a \gls{CRDT} cannot maintain a specific shape or topology,
such as a \gls{DAG}, as concurrent updates could violate invariants.
Specific implementations of \glspl{CRDT} can overcome this restriction however,
for example shown by the \emph{Operation-based Add-only monotonic \gls{DAG}} described in~\cite{Shapiro2011}.
Their implementation allows clients to collaboratively edit a DAG, by adding vertices and edges in separate updates.
Topology preservation is enforced by rejection of new edges that violate the current partial order of the \gls{DAG}.
In a similar vein, the \gls{MEG} is designed in a way that preserves its topology as rooted \gls{DAG} inherently, which we will show in \cref{subsec:preserve_dag}.

\section{Assumptions and Architecture}
\label{sec:system-model}
We assume a finite and known set of replicas, each storing a full local copy of the \gls{MEG}.

%
%
%
%

{\bfseries Assumptions.}
We make use of two failure models, both based on the \emph{asynchronous} timing assumption,
which means that no upper bounds on computation or network transmission times are given.
The \emph{fail-silent} model~\cite[p. 63]{Cachin2011} implies that faulty replicas can crash-stop at any time,
while the remaining replicas have no means to reliably distinguish failure from communication or processing delays,
i.e., the fault is `silent'.
The \emph{fail-silent-arbitrary} model~\cite[p. 64]{Cachin2011} allows for arbitrary,
i.e. byzantine,
behavior of faulty replicas.
This includes intentionally malicious behavior.
In this model,
`silent' also means that replicas cannot detect whether another replica currently adheres to the protocol or not.

We call a replica \emph{correct} if it is non-faulty.
A fault is the failure to adhere to the protocol.
Additionally, in the fail-silent model, a replica is also considered faulty if it is crashing infinitely often, remains crashed forever or looses its memory upon recovery.~\cite{Cachin2011}

The formal CRDT-proof that we give in \cref{sec:mx-crdt} is based on the stricter assumption of a \emph{fail-silent} model.
In \cref{sec:weakening-assumptions} we extend the claims to the \emph{fail-silent-arbitrary} model.

Furthermore, we make use of two broadcast abstractions in this work.
Firstly, we use \emph{Reliable Broadcast}.
Informally, this abstraction provides a set of properties that guarantee that eventually,
the same set of messages is received by all correct replicas,
even if the sending replica fails~\cite{Cachin2011}.

\begin{itemize}
	\item \emph{Validity:} If a correct replica sends a message $m$, then it eventually receives $m$.
	\item \emph{No duplication:} Messages are received only once.
	\item \emph{No creation:} If a replica receives a message $m$ with sender $p$, then $m$ was previously sent by $p$.
	\item \emph{Agreement:} If a message $m$ is received by some correct replica, $m$ is eventually received by every correct replica.
\end{itemize}

The other, more powerful, abstraction is called \emph{Causal Order Reliable Broadcast}.
It extends the guarantees of Reliable Broadcast by also preserving the \emph{causal order} of messages~\cite{Cachin2011}:

\begin{itemize}
	\item \emph{Causal Delivery:} For any message $m_1$ and $m_2$ where the broadcast of message $m_1$ \emph{happened before} (cf.~\cite{Lamport1978})
		the broadcast of message $m_2$, $m_2$ is only received by replicas that have already received $m_1$.
\end{itemize}

The formal CRDT-proof in \cref{sec:mx-crdt} is based on the \emph{Causal Order Reliable Broadcast} abstraction.
In \cref{sec:weakening-assumptions} we relax this assumption to \emph{Reliable Broadcast} --- even in byzantine scenarios --- while maintaining the \gls{CRDT} properties.

{\bfseries Architecture.}
As we can see in \cref{fig:sys_model}, each \emph{client} is attached to a single \emph{replica} in which it trusts.
The client can request functions of class \texttt{query} or \texttt{update} at their replica,
as defined in \cref{subsec:crdt}.
As part of executing an \texttt{update} function, the source replica distributes operations, i.e., encoded state changes,
to all replicas using a broadcast \emph{communication abstraction}.

\begin{figure}
\resizebox{0.95\linewidth}{!}{
\includegraphics[]{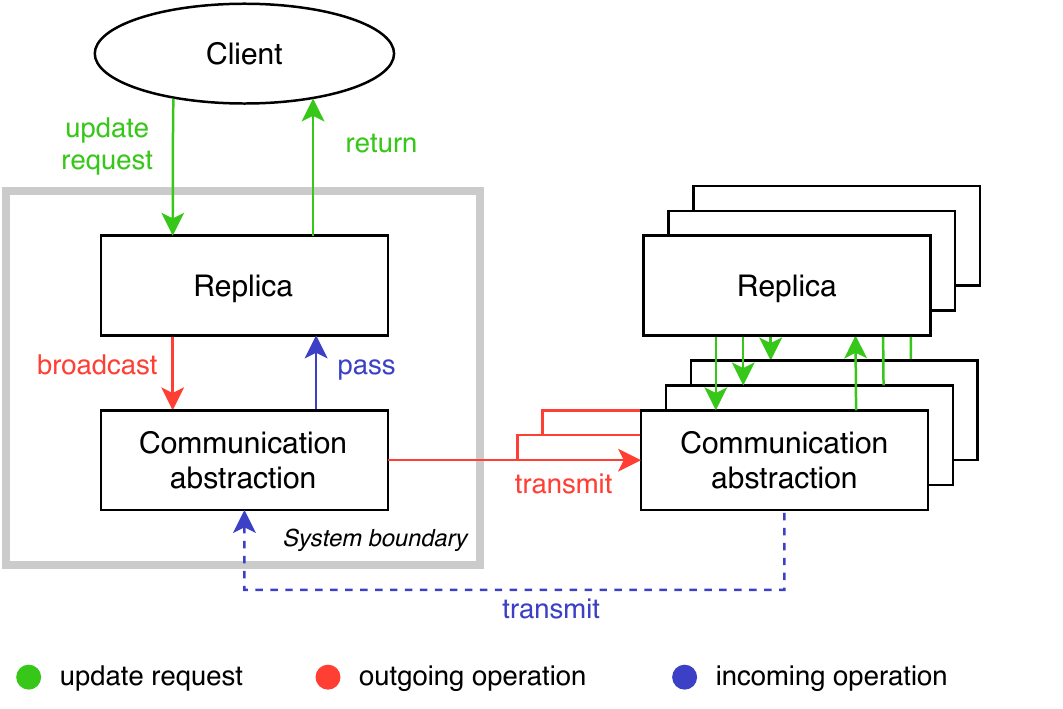}}
\caption{An update request by a client invokes the generator of an update function at the replica,
which creates an update operation.
This update operation is then transmitted to all replicas,
including the calling replica itself, through the communication abstraction.
The communication abstraction enforces guarantees about incoming operations, e.g. on their ordering.
}
\label{fig:sys_model}
\end{figure}

\begin{figure}
\resizebox{0.95\linewidth}{!}{
\includegraphics[]{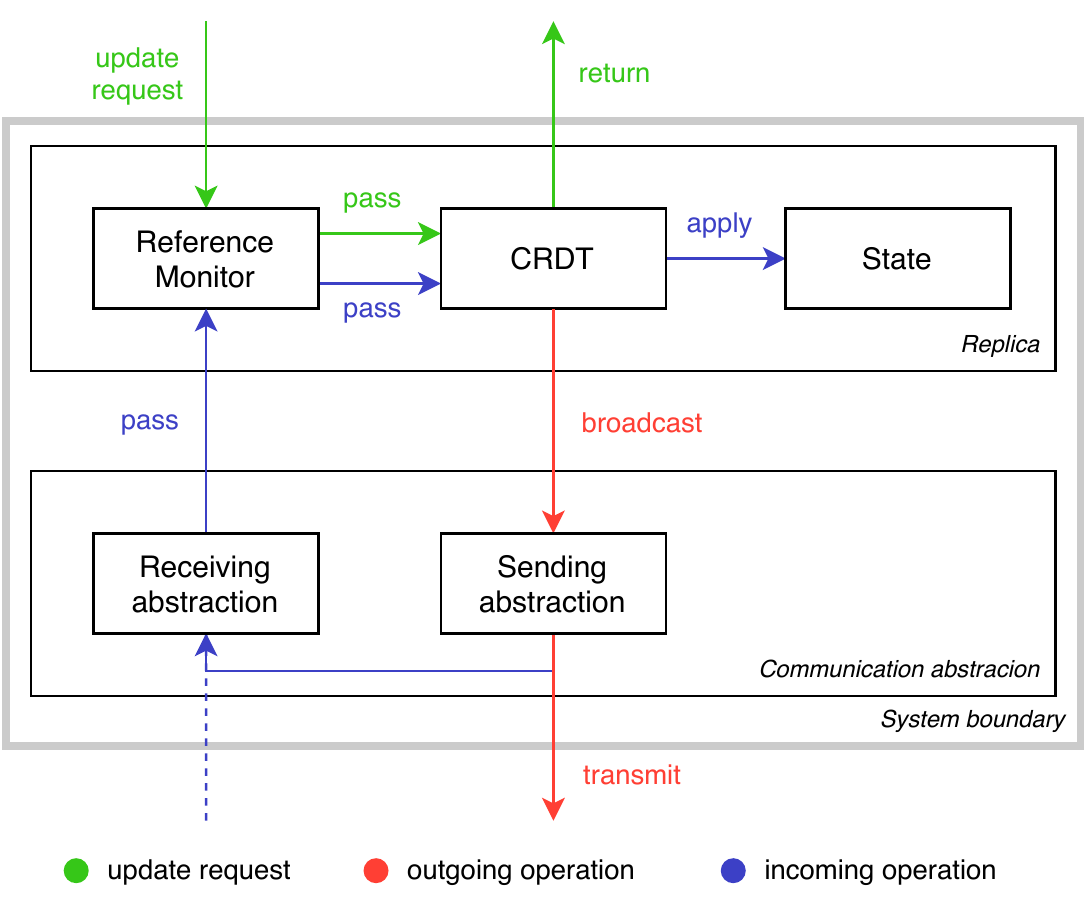}}
\caption{Inner workings of the source replica and communication abstraction when receiving an update request.
After entering the replica through the Reference Monitor, it is passed to the \gls{CRDT}.
The \gls{CRDT} encodes the state changes as an operation which is then broadcasted to all replicas using the communication abstraction.
Incoming update operations, again, pass the Reference Monitor before being processed at the \gls{CRDT} component.
The \gls{CRDT} then applies them to the local state of the replica.
}
\label{fig:inner_view}
\end{figure}

A more granular architectural view is provided in \cref{fig:inner_view}.
Inside a replica, the \emph{Reference Monitor} is the entry point for incoming requests from clients and operations from remote replicas.
It serves as a gate keeper to prevent further processing of operations or requests that violate the protocol or, in a byzantine setting, originate from unauthorized or unauthenticated parties.
Operations and requests that pass the Reference Monitor are handed to the \emph{\gls{CRDT}}.
The \gls{CRDT} can read and modify the state of the replica and is thus the core logic module of the replica.
In case of a \texttt{query} request, it accesses the state and returns the desired value.
For \texttt{update} requests, the generator of the update function encapsulates state changes into an operation that is passed to the communication abstraction.
The \gls{CRDT} then returns to the client to indicate success.
The communication abstraction sends the update operation to all replicas, including the calling replica itself.\footnote{While, depending on the specific communication abstraction, this is not required in an actual implementation, it is important on a conceptual level to ensure that the guarantees hold.}
These update operations then trigger the local update effector which applies the changes to the state of the replica.


\section{The MEG as CRDT}
\label{sec:mx-crdt}
Building upon the overview given in \cref{sec:mx_event_graph}, we formalize the \gls{MEG} as an operation-based shared object.
We show that the \gls{MEG} is a \gls{CRDT}
and thereby provides \acrfull{SEC}.
The underlying assumption for this section is a fail-silent model with Causal Order Reliable Broadcast. This is in accordance with the assumptions used by \citeauthor*{Preguica2011} for \glspl{CRDT} (cf.~\cref{subsec:crdt})~\cite{Preguica2011}.

\subsection{Formalization of the MEG}
\label{subsec:formalization}
To define the Matrix Event Graph as a \gls{CRDT}, we adopt the formal definition introduced with the concept of operation-based \glspl{CRDT} in~\cite{Preguica2011, shapiro-crdt} and use the pseudo code notation by \citeauthor*{Preguica2018}~\cite{Preguica2018}.

An object is formally defined as $(S, s^0, q, t, u, P)$:
$S$ is the space of possible per-replica states,
and $s^0 \in S$ is the \texttt{initial} state of every replica.
$q$ is the set of \texttt{query} functions.
\texttt{update} functions are composed of a \texttt{generator} step $t$ and an \texttt{effector} step $u$.
The \texttt{effector} $u$ may contain a \emph{delivery precondition} $P$,
which must be fulfilled before an operation is being processed further.
Notably, $P$ only delays the execution, it does not abort the \texttt{effector} step.
When a replica with \texttt{state} $s \in S$ executes a step $u$, we denote this as $s \bullet u$, which yields a new \texttt{state}.
As shorthand for the \texttt{state} at replica $i$, we write $s_i \in S$.

We provide a pseudo code implementation of the \gls{MEG} as an operation-based \gls{CRDT} in \cref{lst:pseudocode}.
A vertex is a tuple $(e, w)$ that represents an event in the \gls{MEG}.
$w$ is a unique identifier for the event,
whereas $e$ contains the actual event.
Edges represent a potential causal relationship between child and parent vertex.
The \texttt{state} is a \gls{DAG}, defined through a set of vertices and a set of edges.
Initially ($s^0$), it consists of a single vertex and no edges.
The \texttt{query} functions \texttt{lookup}, \texttt{hasChild}, \texttt{getExtremities} and \texttt{getState}
allow to access the replica state without modification.
\texttt{lookup} checks whether a vertex with a given identifier is part of the current state.
Similarly, \texttt{hasChild} checks for the existence of child vertices for a given vertex.
\texttt{getExtremities} returns the current set of forward extremities, whereas \texttt{getState} returns the \texttt{state}.
The \texttt{update} function \texttt{add} is used to append new events to the \gls{MEG}.
Its \texttt{generator} step $t_{add}$ takes the event $e$ as input argument.
Based on the \texttt{state} of the source replica at that time, a set $L$ of forward extremities is created.
Lastly, a unique identifier $w$ is chosen.
The parameters $w$, $e$ and $L$, and a reference to the update function \texttt{add} are returned together an constitute the update operation.

The \texttt{effector} $u_{add}$ is invoked by the operation that was created in the \texttt{generator} step.
Once the delivery precondition $P$ is fulfilled,
the new vertex $(e,w)$ and the new edges $((e,w),(e_p, w_p))$ for each $(e_p, w_p) \in L$ are added to the \texttt{state},
i.e., the set of vertices and edges, respectively.
Since \texttt{add} is the only update function, we will drop it as a subscript for the steps $t$ and $u$ from now on.

\begin{lstlisting}[mathescape={true}, float,
caption={Pseudo code implementation of the Matrix CRDT.
\texttt{query} and \texttt{update} indicate the type of the respective functions,
\texttt{generator} and \texttt{effector} denote the two steps of an \texttt{update} function.
\texttt{pre} is the delivery precondition $P$.},
label={lst:pseudocode}]
state set $S=(V, E)$ // vertices $\color{gray}V$ consist of event $\color{gray}e$ and uid $\color{gray}w$: $\color{gray}(e,w)$, $\color{gray}E$ are edges: $\color{gray}E \subseteq V \times V$
	initial $(\{(e_0, w_0)\}, \emptyset)$
	query lookup (uid $w$) : boolean
		return $\exists ((e', w') \in V) : w' == w$
	query hasChild (vertex $(e, w)$) : boolean
		return $\exists ((e', w') \in V) : ((e', w'), (e, w)) \in E$
	query getExtremities () : list of vertices
		return $L=\textstyle\bigcup_{(e,w)\in V : \text{ not hasChild(}(e,w)\text{)}} \{(e,w)\}$
	query getState () : set
		return $S$
	update add
		generator (event $e$)
			let $L =$ getExtremities()
			let $w =$ unique()
			return add, ($e, L, w$)
		effector (event $e$, list of vertices $L$, uid $w$)
			pre: $\forall (e_p, w_p) \in L$: lookup($w_p$)
			$V = V \cup \{(e,w)\}$
			$E = E \cup \textstyle\bigcup_{(e_p, w_p) \in L} \{((e,w), (e_p, w_p))\}$
\end{lstlisting}

\subsection{Preservation of the DAG topology}
\label{subsec:preserve_dag}

As mentioned in \cref{subsec:crdt}, the preservation of a specific shape,
such as a \gls{DAG}, is not possible in a generic way for \glspl{CRDT}.
We now show that the \gls{MEG} always preserves the desired data structure of a rooted \gls{DAG} by design as \cref{lm:dag-property}.

\begin{lemma}
\label{lm:atleastoneedge}
  There is at least one forward extremity at any time after initialization of the \gls{MEG}.
\end{lemma}
\begin{proof}
  By induction.\\
  \emph{Base case:} After initialization of the \gls{MEG}, the \gls{DAG} consists of a single root and no edges. Therefore, the root is a forward extremity as it has no incoming edges.\\
  \emph{Induction step:} Given a valid \gls{MEG}, executing \texttt{add} appends a new vertex with only outgoing edges.
  Thus, that new vertex is a forward extremity.
\end{proof}

\begin{lemma}
\label{lm:dag-property}
  The \gls{MEG} maintains the properties of a rooted DAG at all times: (i) single root, (ii) acyclicity, and (iii) weak connectedness.
\end{lemma}
\begin{proof}
By induction. \\
\emph{Base case:} The initial state $s^0$ contains a single vertex and no edges.
This \gls{MEG} therefore is a rooted DAG.\\
\emph{Induction step:} Given replicas $i$ with state $s_i=(V_i,E_i)$, where $s_i$ is a rooted DAG, an arbitrary source replica $r$ is selected.
As part of the \texttt{generator} step $t$, the set of forward extremities is determined as $L$, and a unique identifier $w$ created.
By \cref{lm:atleastoneedge}, $|L|>0$.
Since $t$ is side-effect-free, the \gls{MEG} remains unchanged.

Consequently, the execution of the \texttt{effector} step $u$ is triggered at each replica $i$.
$u$ awaits the fulfillment of the delivery precondition $P$, which ensures that $s_i$ contains all parents that are referenced by $L$.
Finally, applying $u$ yields the new replica states $s_i'$:
\begin{align*}
  s_i'= (V_i \cup \{(e,w)\}, E_i \textstyle\bigcup_{(e_p, w_p) \in L} \{(e,w), (e_p,w_p))\}).
\end{align*}
Since all new edges are outgoing from the new vertex $(e,w)$, no new cycles can be formed, and existing roots remain roots.
No new roots or isolated vertices have been added as the new vertex has outgoing edges.
Because all $s_i$ were assumed to be rooted \glspl{DAG}, all $s_i'$ must be rooted \glspl{DAG}.
\end{proof}

\subsection{Proof of CRDT properties}
\label{subsec:crdt_proof}
Now, we show that \gls{MEG} implements an operation-based \gls{CRDT}
and thus guarantees \gls{SEC}.
We structure the proof by the \gls{SEC} properties Strong Convergence, Eventual Delivery, and Termination (cf. \cref{subsec:consistency}).

{\bfseries Strong Convergence.}
For Strong Convergence, we need to show commutativity of concurrent updates
and causal order reception of operations for noncommutative updates.

Commutativity for updates is determined by the commutativity of their operations.
Two updates ($t, u$) and ($t', u'$) commute,
iff for any reachable state $s \in S$ for which the delivery precondition $P$ is satisfied for both $u$ and $u'$:
(i) $P$ is still satisfied for $u$ in $s \bullet u'$, and
(ii) $s \bullet u \bullet u' \equiv s \bullet u' \bullet u$.~\cite{shapiro-crdt}

\begin{lemma}\label{lm:P_remains}
Once an update operation satisfies $P$ for some state $s$, it will continue to satisfy $P$ for any state $s'$ following $s$.
\end{lemma}
\begin{proof}
  Consider any update operation $u(e,L,w)$ that satisfies $P$ in some state $s=(V,E)$.
Applying an arbitrary operation $u(e',L',w')$ to $s$ yields a new state $s'$:
\begin{align*}
s' &= s \bullet u(e', L', w') \\
  &= (V \cup \{(e',w')\}, E \cup \textstyle\bigcup_{(e_p, w_p) \in L'}\{(e', w'), (e_p, w_p)\})
\end{align*}
$P$ being satisfied in $s$ implies that it remains satisfied for $s'$:
\begin{align*}
  &\forall (e_p,w_p) \in L : (e_p, w_p) \in V \\
  \Rightarrow &\forall (e_p,w_p) \in L : (e_p, w_p) \in V \cup \{(e',w')\}
\end{align*}
\end{proof}

\begin{lemma}\label{lm:crdt-com}
Any two operations $u(e_i, L_i, w_i)$ and $u(e_j, L_j, w_j)$ commute with each other.
\end{lemma}
\begin{proof}
We consider any state $s=(V,E)$ and two update operations $u(e_i, L_i, w_i)$, $u(e_j, L_j, w_j)$ that both satisfy $P$ in $s$.

As shown in \cref{lm:P_remains}, after applying one operation, the other operation still satisfies $P$.
It remains to show that the resulting states are equivalent, regardless of the order in which the effectors are executed.
%
%
Since $u$ only performs a union of the edge and vertex sets, by commutativity of the union operator,
commutativity of $u$ follows:
%
$s \bullet u(e_i, L_i, w_i) \bullet u(e_j, L_j, w_j)
  \equiv s \bullet u(e_j, L_j, w_j) \bullet u(e_i, L_i, w_i)$
\end{proof}


As we have shown, \gls{MEG} updates are commutative and Strong Convergence is guaranteed.
This is possible because all required properties of the \gls{MEG} are preserved by design (cf.~\cref{lm:dag-property}).

Since the \gls{MEG} encodes causal relations as edges in the data structure,
the delivery precondition $P$ can ensure that these dependencies are respected without sacrificing commutativity.

{\bfseries Eventual Delivery.} For Eventual Delivery, we need to show that $P$ is eventually satisfied for all operations.

\begin{lemma}\label{lm:P_satisfied}
$P$ is immediately satisfied on causally ordered message reception.
\end{lemma}
\begin{proof}
$P$ ensures that all referenced parents are part of the local \texttt{state}.
Since \texttt{getExtremities} selects all parents from the current \texttt{state}, $P$ must be satisfied at the source replica after the \texttt{generator} step.
Once satisfied, $P$ remains satisfied since vertices are never removed.
Therefore, receiving all causally preceding operations is sufficient to satisfy $P$ at every replica.
Consequently, having causal order message reception, $P$ is immediately satisfied on reception.
\end{proof}

{\bfseries Termination.}
Given the implementation in \cref{lst:pseudocode}, we can see that there are no loops or recursive calls in either of the functions, therefore, they will eventually exit.
Knowing that $P$ is immediately satisfied given causal order message reception, as shown in \cref{lm:P_satisfied}, we can conclude that Termination holds.

{\bfseries Conclusion.}
We have shown Termination and eventual satisfaction of $P$.
\cref{lm:crdt-com} shows commutativity of concurrent updates.
Therefore, all properties of an operation-based \gls{CRDT} are met by the \gls{MEG}.

\section{Relaxation of Assumptions and Reality Check for Byzantine Settings}
\label{sec:weakening-assumptions}
In this section, we evaluate the assumptions we have used for the \gls{CRDT} proof of the \gls{MEG} in \cref{sec:mx-crdt} and relax them wherever possible without violating previously shown guarantees.
We show that Matrix currently provides no \gls{SEC} because of its unreliable broadcast protocol.
However, when having a Reliable Broadcast abstraction that provides Validity and Agreement,
the \gls{MEG} can provide \gls{SEC} in byzantine $n > f$ environments
with $n$ total and $f$ faulty participants.
This is possible since conflicts, created by byzantine replicas that share different update operations with different replicas, can always be resolved.

\subsection{Relaxation of the Broadcast Assumptions}
\label{subsec:weak_comm}

In \cref{sec:mx-crdt}, we assumed a Causal Order Reliable Broadcast abstraction,
which is commonly used with \glspl{CRDT}.
Yet in reality, the communication abstraction employed by Matrix provides much weaker guarantees.
We thus revisit the assumptions
and show that the Causal Delivery property of the broadcast abstraction is not necessary\footnote{The No Duplication
property is also not necessary: Because each vertex has a unique identifier $w$, and outgoing edges cannot be added
afterwards, it suffices to make the effector conditional on the presence of the vertex in the replica state to gain
idempotent effectors that can cope with multiple receptions of identical operations.}
and can be removed without violating Strong Convergence for safety as well as Eventual Delivery and Termination for liveness (cf. \cref{sec:background} for definition and \cref{sec:mx-crdt} for fulfillment).

{\bfseries Strong Convergence.}
To provide Strong Convergence, replicas must receive noncommutative update operations in their causal order.
As every update operation commutes with every other, as shown in \cref{lm:crdt-com},
Strong Convergence does not require any ordering guarantees by the communication abstraction.

{\bfseries Eventual Delivery.}
In \cref{lm:P_satisfied}, we used the Causal Delivery property to show that the delivery precondition $P$ is immediately satisfied.
However, Eventual Delivery only requires that correct update operations received by a replica \emph{eventually} satisfy $P$,
so that they can eventually be applied.

It therefore remains to show that the delivery precondition $P$ is eventually satisfied without Causal Delivery.
Given an update operation, $P$ is satisfied if all referenced parents are part of the \texttt{state} of a replica.
If an operation satisfies $P$ at some point in time,
it continues to satisfy $P$ thereafter,
because the \gls{MEG} is an append-only data structure.
As per \cref{lm:P_satisfied}, $P$ is satisfied for any given operation after the \texttt{generator} step at the source replica finishes.
Therefore, all referenced parents must have been previously added to the \texttt{state} and therefore be part of some update operation.
If an update operation does not satisfy $P$ at some replica due to reordering of operations by the broadcast abstraction,
replicas can delay and buffer the update operation until $P$ is satisfied.
Owing to the Validity and Agreement properties of the broadcast abstraction (cf. \cref{sec:system-model}), all missing update operations are eventually received by all correct replicas.
As correct replicas apply all operations that they received and that satisfy $P$,
all parents must eventually be part of their \texttt{state}.
Consequently, for correct replicas,
$P$ must eventually be satisfied for every update operation.

{\bfseries Termination.}
Since all method executions terminate, and since we have shown that in the new setting, $P$ is eventually satisfied for all operations, the Termination property still holds.

Thus,
the \gls{MEG} only requires a weak form of Reliable Broadcast, and does not depend on Causal Delivery.

\subsection{Tolerating Byzantine Failures}
\label{subsec:byzantine}
In the following, we replace the \emph{fail-silent} failure model with the \emph{fail-silent-arbitrary} model.
We assume that the adversary cannot permanently block broadcast communication between two correct replicas.
In a system with $n$ replicas,
the adversary can induce byzantine faults in up to $f$ replicas
with $n > f$.
This means that a client's trusted replica might be the only correct replica in the system.
As the \gls{MEG} does not strive for consensus,
it is able to cope with such a hostile environment.
To model the capabilities of byzantine replicas in a distributed systems that implement a \gls{CRDT},
\citeauthor*{Zhao2016} introduce a three-part threat model~\cite{Zhao2016}
which consists of attacks on the membership service, malicious updates, and attacks on the Reliable Broadcast service.
To keep focus on the \gls{MEG}, we will only touch on the issues related to the membership service and malicious updates,
and put the attack on the Reliable Broadcast service at the center of attention.

{\bfseries Membership service.} With respect to a membership service, we assume a known set of replicas that does not change.
Still, we want to note that attacks on the membership service for dynamic groups may prevent replicas from receiving some or all update operations, which could affect Eventual Delivery.
We consider this as an important, but somewhat separate topic.

{\bfseries Malicious updates.} Malicious replicas could attempt to inject updates into the data structure that are not compliant with the protocol.
In general, to address threats from malicious updates,
the Reference Monitor is the endpoint for all external interfaces of the replica.
It ensures authorization,
authentication, integrity, and general protocol compliance of incoming operations.
Update operations that pass the Reference Monitor can therefore be handled like non-byzantine,
i.e., correct operations.
A serious attack could be based on non-unique event identifiers.
However, unique event identifiers can be ensured in a byzantine environment by generating event identifiers from the event data using a collision-resistant hash function.
This way, Reference Monitors can verify whether an event identifier is valid by recomputing the hash themselves.
To prevent the injection of unauthorized update operations,
impersonation needs to be prevented as well.
This can be achieved by means of asymmetric encryption, i.e., by cryptographically signing update operations and a Public Key Infrastructure that is trusted by all correct replicas.
Signatures also ensure integrity of update operations,
so that update operations that are not directly received from the source replica cannot be altered unobtrusively.
Therefore,
authenticated update operations allow us to drop the No Creation property of the broadcast abstraction,
as the Reference Monitor can now identify forged or tampered update operations itself.
The creation of operations that are not protocol compliant,
such as events with non-existing or non-(con)current extremities as parents,
might incur load on performance, but does not threaten the correct operation of the \gls{MEG}.

{\bfseries Attacks on the Reliable Broadcast Abstraction.}
Attacks on the Reliable Broadcast abstraction may lead to correct replicas that receive different operations,
potentially causing permanent divergence in replica states.
While fail-silent-arbitrary Reliable Broadcast algorithms exist (cf.~\cite[p. 121]{Cachin2011}),
they are generally difficult to scale to many replicas,
as communication complexity increases in the number of replicas.
However, we do not require all of their properties due to the commutative and conflict-free nature of the \gls{MEG}.
Using asymmetric encryption, the No Creation property is not required and the broadcast abstraction is left to provide Validity and Agreement.
As Validity is only concerned with correct sending replicas,
faulty replicas can mainly attack Agreement by performing equivocation,
i.e. broadcasting different update operations to different replica subsets,
or not broadcasting an update operation to all replicas~\cite{non-equivocation}.
We show that equivocation,
a costly problem in fail-arbitrary Reliable Broadcast algorithms,
is not an issue for the Matrix Event Graph due to its distinct structure.
We recall that for Agreement,
an operation that is received by some correct replica eventually has to be received by every correct replica.

Under the assumption that malicious replicas have no means to fabricate a hash collision,
they can only send operations with different event identifiers when trying to create inconsistencies.
However, due to the conflict-free nature of an operation-based \gls{CRDT},
both operations can be received and processed by correct replicas.
A byzantine replica that performs equivocation
can therefore be modeled as two replicas that crash while sending independent update operations.
Therefore, the broadcast abstraction only has to ensure that eventually,
\emph{any} operation received at some correct replica will be received at every correct replica.

In Matrix, Validity is provided since source replicas immediately apply update operations to their local state.
However, with respect to Agreement, Matrix replicas use a `best-effort broadcast' that is implemented via unicast transmissions to all replicas.
This alone does not provide Agreement even in fail-silent systems without byzantine attackers,
as a failing replica could only provide a limited number of correct replicas with the update operation.
To mitigate this issue, Matrix uses a backfilling mechanism
which allows replicas to specifically request missing operations from other replicas.
It is used when a replica receives an update operation for which the parents are not part of the replica state.
With this mechanism,
Matrix achieves Agreement under the assumption of constant \gls{MEG} progress,
i.e., a never-ending stream of (arbitrary low-frequent) new update operations from other replicas.
However, if / for as long as the progress come to a halt, Agreement, and thus Eventual Delivery,
is violated\footnote{In the Matrix reference replica implementation Synapse,
this issue has been raised in the developer community~\cite{githubdelivery}.
Correct replicas will now take note of unreachable homeservers
and retry synchronization once they become available eventually~\cite{githubcatchup}.
Faulty senders still require constant progress.
}.

Therefore, Matrix does only provide Agreement and thereby \gls{SEC} under the assumption of constant progress.
One could now replace the best-effort broadcast with a gossip-based broadcast protocol that is scalable and robust,
as suggested in~\cite{glimpseofthematrix}.
While this alone is not sufficient to ensure Agreement without constant progress,
the efficient gossip-based broadcast
could be used by replicas to periodically broadcast their current set of forward extremities to all other replicas,
which then could trigger backfilling.
This addition would guarantee probabilistic Agreement, and therefore \gls{SEC} for the \gls{MEG} implementation of Matrix.

\section{Scalability: Width of the MEG over Time}
\label{sec:convergence}
In this section, we study the evolution of the width of the \gls{MEG} over time.
While we verified our results with Monte-Carlo simulations,
we decided to go for an analytical approach to deliver a precise mathematical problem definition and treatment.
In Sections \ref{sec:mx-crdt} and \ref{sec:weakening-assumptions},
we assumed that \textit{all} forward extremities known to a replica are used as parents for new vertices created by the replica.
In this case, the number of forward extremities is reduced as much as possible whenever a new vertex is created.
However, as noted in \cref{sec:mx_event_graph},
honest replicas can experience a high number of forward extremities after a partition,
and malicious replicas could deliberately create events with a high number of parents.
This is problematic from a performance perspective because checks, particularly of the Reference Monitor, are resource intensive,
especially when old parts of the \gls{MEG} are referenced,
but are needed for every parent~\cite{synapse-issue-forward-extremities-accumulate}.
Thus, for reasons of performance, the number of parents of a new vertex
is restricted to a finite value $d$ in practice.
If there are more than $d$ forward extremities,
a replica selects a random subset of parents of size $d$ for the new vertex.
In this section, we provide evidence that the width of the \gls{MEG} still converges\footnote{Please note that when we discuss convergence in this section, convergence is related to the number of forward extremities. In the previous CRDT-related section, convergence is related to propagation of states.} to the the number $k$ of participating replica times a small factor when all $k$ replica repeatedly and concurrently add a new vertex.



We model the evolution of the width of the \gls{MEG} as follows.
We assume that vertices are added in rounds.
A round consists of two steps:
First, each of the $k$ replicas concurrently adds a new forward extremity and thereby `eliminates' $d$ forward extremities which are used as parents.
Second, all replicas synchronize their new extremities and reach a consistent state.
The overall number of eliminated extremities depends on the amount of \textit{overlap} between the parent choices of different replicas.
As we are interested in scaling $k$ while keeping $d$ low,
we assume $k$ > $d$. As forward extremities cannot be eliminated effectively if a new forward extremity has only one parent, we assume $d > 1$.
The model also accepts an arbitrarily high number of forward extremities $u_0$ as starting condition.
We analyze the sequence of number of forward extremities $u_i$ by a mean value analysis.

Please note that this model maximizes uncoordinated concurrency in Step 1 and, thus,
models a worst case scenario:
More new vertices per replica in Step 1,
i.e., a higher frequency of updates by clients
or prolonged periods of network partition,
would eliminate more than $d$ overlap-free forward extremities,
but not add additional ones.
Also, if replicas would be aware of the eliminations of other replicas,
their forward extremity choices could be done more overlap-free.


\subsection{Stochastic Process}\label{subsec:stochastic_process}

We represent the concurrent updates in Step 1 of each round as a stochastic urn model.
The initial number of forward extremities $u$ is described by $u$ initial red balls,
while the number of newly linked parent vertices $d$ is the number of balls taken out during a drawing by a replica.
The update generator execution of the $k$ replicas lead to the conduction of $k$ independent drawings that can be modeled by sequential drawings with the use of black balls: the balls drawn by a replica are replaced by black balls and put back to the urn.
Therefore, after $k$ replicas have performed Step 1, the black balls indicate the number of selected parent vertices.
After each round, the black balls are replaced by red ones again and the next round starts with the current number of red balls.

We let the random variable $R_{d,k}(u)$ denote the total number of removed forward extremities,
while $u - R_{d,k}(u)$ denotes the number of forward extremities that `survived' for the subsequent urn experiment.
With this urn experiment,
we build a stochastic process for the behavior of the number of forward extremities.
We derive the expectation and the variance of $R_{d,k}(u)$, and we provide a recursion formula for
the distribution of $R_{d,k}(u)$.
We discuss the implications on \glspl{MEG} in \cref{subsec:implications}.

Let the random variable $U_n$ describe the number of balls in the urn after $n \in \mathbb{N}_0$ rounds.
Let $u_0$ be the initial number of balls in the urn, then $U_0 = u_0$ and $U_{n+1} = U_n + k - R_{d,k}(U_n)$.
As $(U_n)_{n \in \mathbb{N}_0}$ is a sequence of random variables, it is a stochastic process (cf.~e.g.~\cite{gallager}).
We are interested in whether convergence can be \emph{expected}, and, if yes, how fast convergence is reached.
The process is a spatially inhomogeneous random walk,
specifically a time-homogeneous Markov chain (cf.~e.g.~\cite{mitzenmacher_upfal}) with state space $M_U = \mathbb{N}^+$:

\vskip-10pt
\begin{eqnarray*}
    \forall n \in \mathbb{N}_0 \enspace \forall u_0, \dots, u_{n+1} \in M_U: \quad\\
                \mathbb{P}(U_{n+1} = u_{n+1} | U_0 = u_0, \dots, U_{n-1} = u_{n-1}, U_n = u_n) \\
\nonumber       = \mathbb{P}(U_{n+1} | U_n = u_n)
                \Rightarrow \text{memorylessness}
\end{eqnarray*}

with transition matrix:
    $P_{i,j} = \mathbb{P}(U_n = j | U_{n - 1} = i) = \mathbb{P}(R_{d,k}(i) = k - (j - i))$
    and transition probability:
    $\forall n \in \mathbb{N}_0 \forall l, m \in M_U: \mathbb{P}(U_n = j | U_{n-1} = i) = \mathbb{P}(U_1 = j | U_0 = i) $.
    Thus, the transitions are independent of $n$ and the process is time-homogeneous.

A positive recurrent, aperiodic and irreducible Markov chain
has a stationary distribution,
i.e., a fixed point of the transition function in which the probabilities for the next state do not change with state transitions.

If we assume $u_0 \in [0, k-1]$,
then $u_1 > k$,
as no more than $u_0$ balls can be drawn,
but $k$ balls get added.
Therefore, states $[0, k-1]$ are transient, and one can remove them from the chain.
The remaining states are irreducible
and aperiodic:
As the next state increment in one round is in $[k - k \cdot d, k - d]$,
every other state can be reached in a finite number of iterations.
However, it is unclear whether the states are transient, i.e., visited only once,
or positively recurrent,
i.e., have a finite expected time until they are visited repeatedly.
This represents an open problem and is left for future work.

\subsection{Properties of Random Variable \texorpdfstring{$R_{d,k}(u)$}{R}}\label{subsec:stochastic_variable}
As stated before,
let $R_{d,k}(u)$ denote the total number of red balls that showed up in a single round of $k$ independent drawings of size $d$ from an urn of size $u$.
Initially, the urn contains only red balls ($r = u$) and no black balls ($b$ = 0).
A {\em drawing} means taking $d$ balls from the urn at random,
where $d < u$.
The drawing ends by replacing each red ball with a black ball and then returning all $d$ balls back into the urn.

We now provide the expectation (a)
and the variance (b) of $R_{d,k}(u)$,
and a recursion formula (c) for the distribution of $R_{d,k}(u)$. For the proof, see \cref{apx:proof}.

\begin{theorem}For the random variable $R_{d,k}(u)$, we have:
\begin{enumerate}
\item[a)] $
\displaystyle{
    \BE(R_{d,k}(u)) = d \cdot \frac{1-p^k}{1-p}, \quad k \ge 1,
}$\\
where
\begin{equation}\label{defp}
p = \frac{u-d}{u}
\end{equation}
is the retention probability.

\item[b)]
\begin{eqnarray*}
\hspace{-8mm}
\BV(R_{d,k}(u)) & = & \frac{vud}{1-p} \left(\frac{1-w^{k-1}}{1-w} - p^{k-1} \cdot \frac{1-(w/p)^{k-1}}{1-w/p}\right) \\
& & - \frac{vd^2}{(1-p)^2} \left(\frac{1\! -\! w^{k-1}}{1-w} - 2p^{k-1} \frac{1\! -\! (w/p)^{k-1}}{1-w/p} \right. \\
& & \hspace{2.5cm} \left. + p^{2(k-1)} \frac{1\! -\! (w/p^2)^{k-1}}{1-w/p^2}\right),
\end{eqnarray*}
where
\begin{equation}\label{defv}
v = \frac{d(u-d)}{u^2(u-1)}, \qquad w = \frac{(u-d)(u-d-1)}{u(u-1)}.
\end{equation}

\item[c)] If $k \ge 2$ then
\[
\hspace{-6mm}
    \PP(R_{d,k}(u)=j) = \sum_{\ell =0}^d \frac{{\binom{u -(j-\ell)}{\ell}}{\binom{j - \ell}{d-\ell}}}{{\binom{u}{d}}} \cdot \PP(R_{d,k-1}(u) = j-\ell).
\]
\end{enumerate}
\end{theorem}

\subsection{Implications for the MEG and Conjecture}\label{subsec:implications}

The formula for the expectation of $R_{d,k}(u)$ allows for statements on the expected convergence behavior of the \gls{MEG}
in the presence of concurrent updates by different replicas.
In addition,
the formula for the variance of $R_{d,k}(u)$ shows the deviation from expected convergent behavior.
For \Cref{fig:forward-extremities-development},
we use these formulas to calculate
the expected development and deviation of forward extremities $U_n$ over the number of rounds for varying $k$ but fixed $d$.
To plot the calculations,
we put different realizations of $U_n$ against the expected value of $U_{n+1}$,
via
$\mathbb{E}(U_{n+1}) = U_{n} + k - \mathbb{E}(R_{d,k}(U_{n})$.
The dashed line is $U_{n+1} = U_n$,
so its intersection with the colored lines mark their fixed points.
In the area below the dashed line,
$\BE(U_{n+1}) < U_n$, the urn contents are expected to decrease, in accordance with the plotted standard deviation.
The change from linear to constant curves (for decreasing $U_n$, i.e.~from right to left) show the switch from likely overlap-free choices to overlapping choices,
which decrease the urn contents less.
It shows that for any plotted realization of $U_n$,
we either expect a decreasing urn value (below the dashed line),
or a transition to the fixed point.
Therefore, the plotted configurations show convergence.
In addition, the variance is very low.
We observe that the convergence of the width of the graph appears to be almost optimal,
i.e., the fixed point is near $k$.

\begin{figure}[tbp]
    \resizebox{\linewidth}{!}{
        \includegraphics[clip,trim=0mm 0mm 0mm 0mm]{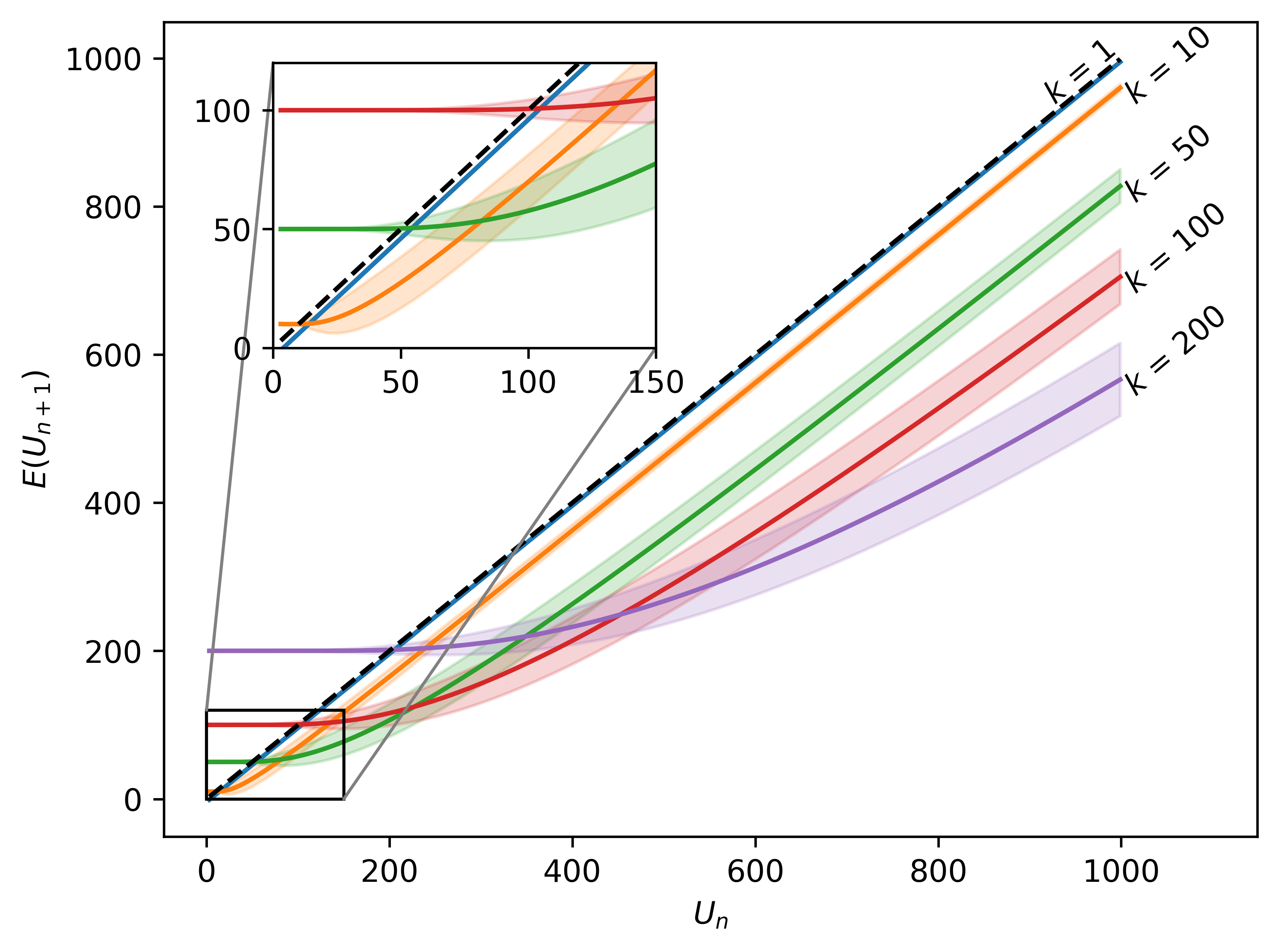}
    }
    \caption{Expectation for the next urn content $\mathbb{E}(U_{n+1})$ for different realizations of $U_n$, $d=5$, and varying $k$. Points below the dashed line of $U_n = \BE(U_{n+1})$ mean that the urn content is expected to decrease, points above mean that an increase is expected. For visibility, the plotted standard deviation is increased by the factor 5. Please note that when the curves are followed from right to left, they change from a linear slope to a constant value close to $k$.}
    \label{fig:forward-extremities-development}
\end{figure}

Synapse, the reference implementation of a Matrix replica,
recently activated a feature to force the depletion of forward extremities by sending
empty `dummy' events using the same parent selection rules as regular events\footnote{Note that Synapse actually takes 5 random forward extremities and 5 of the newest forward extremities, which are not independent between replicas.} with $d=10$,
as soon as there are more than 10 forward extremities present~\cite{synapse-dummy-events}.
This fact allows to take advantage of the convergence in periods of missing updates,
and brings reality closer to our model.

\begin{figure}[tbp]
    \resizebox{\linewidth}{!}{
        \includegraphics[clip,trim=0mm 15mm 0mm 10mm]{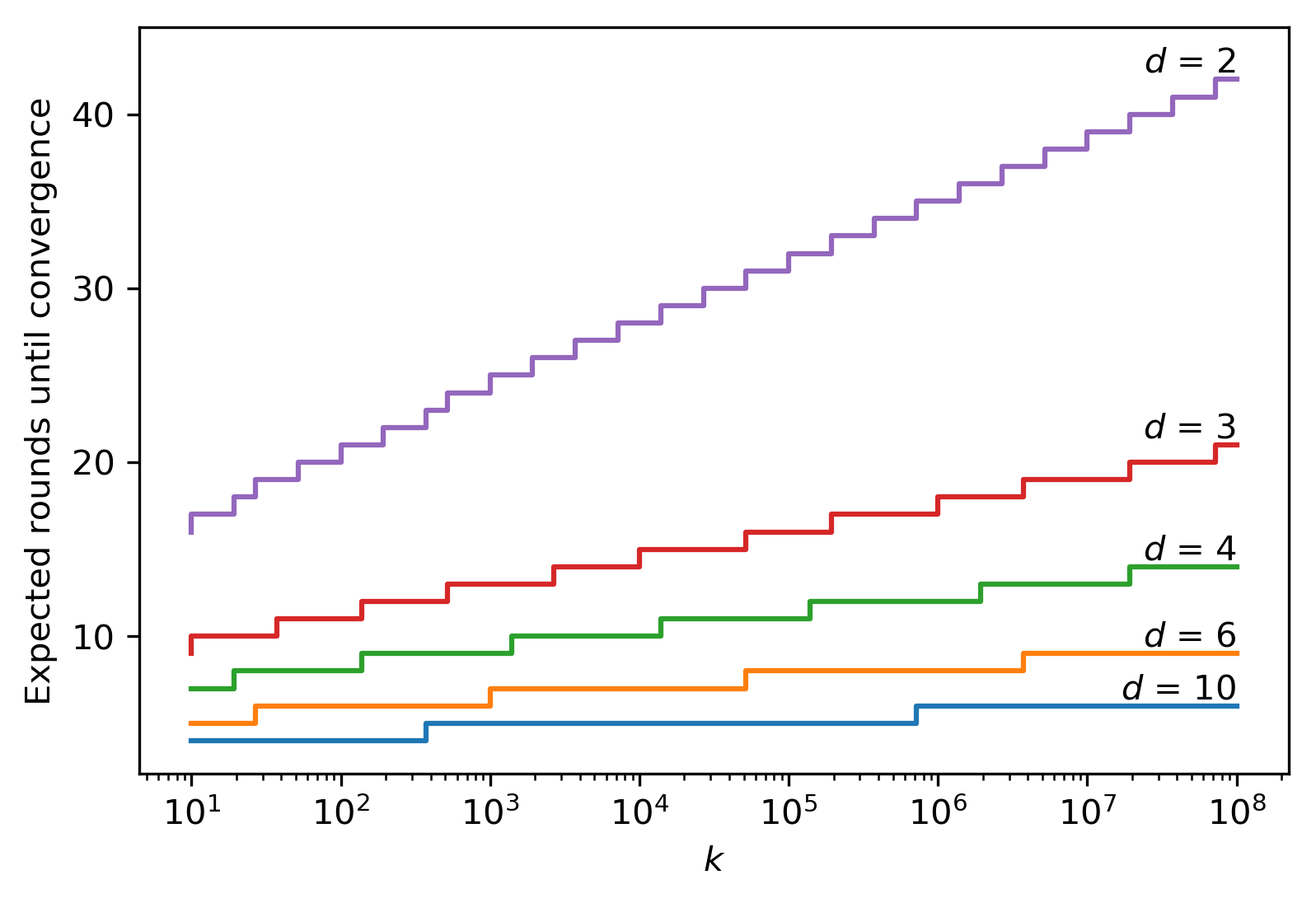}
    }
    \caption{Expected number of rounds until convergence for varying $d$ and $k$, starting at $u_0 = 100\cdot k$. While convergence speed increases with $d$, the returns in the number of rounds to reach convergence diminish.}
    \label{fig:rounds-until-convergence}
\end{figure}

To gain insights into the influence of $d$,
we use the expectation of $U_n$ via $\BE(U_{n+1}) = \BE(U_n) + k - \BE(R_{d,k}(U_n))$,
and calculate the number of rounds $n$ until $\mathbb{E}(U_n) - \mathbb{E}(U_{n+1}) < 1$.
This is equivalent to the number of rounds after which $\mathbb{E}(R_{d,k}(\mathbb{E}(U_n))) \ge k$ holds, i.e., the number of rounds after which we expect to eliminate a number of forward extremities in Step 1 that is less than or equal to the number of forward extremities that we add in Step 2.
\cref{fig:rounds-until-convergence} shows that, while the number of rounds until convergence is reached directly depends on the choice of $d$,
there are diminishing returns.
The highest gain in time until convergence is between $d=2$ and $d=3$,
while there is much less difference between $d=6$ and $d=10$.
With optimal choice of forward extremities, i.e., $u \gg k\cdot d$, convergence speed is nearly $k \cdot (1-d)$,
and therefore the number of rounds until convergence is nearly proportional to $\frac{1}{1-d}$.
Synapse employs $d=5$ with $k \lessapprox 10^3$,
which we can confirm as a good compromise in convergence speed performance using our formulas in
\cref{fig:forward-extremities-development,fig:rounds-until-convergence}.

With small $u$, bad choices,
i.e., overlapping choices for parents are made,
but because $u$ is small,
they don't harm convergence permanently.
With large $u$, the probability for overlapping choices grows smaller and smaller, and convergence speed is linear.

We therefore conjecture that regardless of the exact choice of $k$ and $d$,
the process converges for any start value $u$ to a stationary value near $k$ in a finite number of rounds.
The derived properties of $R_{d,k}(u)$ are important building blocks to eventually prove this conjecture.
The convergence speed depends on the choice of $d$,
but values larger than 3 are subject of diminishing returns.

In practice, this means that if the conjecture holds, the \gls{MEG} possesses a self-stabilization property~\cite{self-stabilization}
in the sense that if transient faults lead to a high number of forward extremities (a high $u$),
a correct system converges to a stable number of forward extremities near $k$ in a finite number of rounds,
and remains stable as if the fault had never occured.

\section{Conclusion}
\label{sec:conclusion}
In this paper, we extracted and abstracted the replicated data type employed by Matrix,
and proved that it represents a Conflict-Free Replicated Data Type.
Therefore, the Matrix Event Graph provides Strong Eventual Consistency, a fact that in particular indicates that all correct replicas that applied the same set of updates are in equivalent state --- immediately and without any further agreement procedure.
This proof gives fundamental insights into why the Matrix system shows good resilience and scalability in the number of replicas in practice.
It therefore makes the underlying replicated data type an attractive candidate as a basis for other decentralized applications.
In addition, we analyzed the challenges for systems with byzantine actors
and showed that the properties of the Matrix Event Graph facilitate a byzantine-tolerant design,
especially due to equivocation tolerance.
However, design and analysis of an appropriate underlying broadcast protocol with the identified properties remain topics for future research.
Furthermore, we formalized and studied the evolution of the width of the graph as a spatially inhomogeneous random walk.
Our observations let us conjecture that the width of the graph always converges independently of the specific system parameters, and does so fast.

In summary, we believe that the Matrix system and similar systems are highly relevant in real-world scenarios, and
that their scientific understanding is of utmost importance.  We hope that our results advance understanding as well
as proper real-world setup of those systems, and can serve as a basis for further research.

\appendix
\section{Proof of Properties of \texorpdfstring{$R_{d,k}(u)$}{R}}
\label{apx:proof}
For a series of drawings $R_{d,k}(u)$,
we write $Z_k$ for the number of red balls that show up in the $k$th drawing,
so that $R_{d,k}(u) = Z_1 + \ldots + Z_k$. \\

a)
In what follows, let $k \ge 2$.
Under the condition $R_{d,k-1}(u) = r$, the urn contains $u-r$ red and $r$ black balls.
Thus, the conditional distribution of $Z_k$ given $R_{k-1} = r$
is the hypergeometric distribution Hyp$(d,u-r,r)$,
which implies
\[
\BE(Z_k|R_{d,k-1}(u) = r) = d \cdot \frac{u-r}{u}.
\]
Since $R_{d,k}(u) = R_{d,k-1}(u) + Z_k$,
we have\\
$\BE(R_{d,k}(u)) = \BE(R_{d,k-1}(u)) + \BE(Z_k)$.
Moreover,
\begin{eqnarray*}
\BE(Z_k) = \BE\big{[} \BE(Z_k|R_{d,k-1}(u))\big{]}
= \BE\bigg{[} d \cdot \frac{u-Z_{k-1}}{u} \bigg{]}\\
= d - \frac{d}{u} \cdot \BE(R_{d,k-1}).
\end{eqnarray*}
It follows that
\begin{eqnarray*}
\BE(R_{d,k}(u)) & = & \BE(R_{d,k-1}(u)) + d - \frac{d}{u}\cdot  \BE(R_{d,k-1}(u))\\
                & = & d + p\, \BE(R_{d,k-1}(u)).
\end{eqnarray*}
Together with $\BE(R_{d,1}(u)) = d$, we now obtain by induction over $k$
\[
\BE(R_{d,k}(u)) = d \, \sum_{j=0}^{k-1} p^j = d \cdot \frac{1-p^k}{1-p},
\]
as was to be shown.
Notice that
\[
\lim_{k \to \infty} \BE(R_{d,k}(u)) = \frac{d}{1-p} = u.
\]
This result is not surprising, sincs in the long run each of the red balls will have shown up.\\

b)  The proof uses the general fact that,
for random variables $X$ and $Y$, the variance of $X$ can be calculated
according to the formula $\BV(X) = \BE\left[ \BV(X|Y)\right] + \BV( \BE[X|Y])$,
i.e., the variance of $X$ is the sum of the expectation of the conditional variance of $X$ given $Y$
and the variance of the conditional expectation of $X$ given $Y$.
In our case, we put $X=R_{d,k}(u)$ and $Y=Z_{k-1}$, where $k \ge 2$,
and obtain
\begin{eqnarray} \nonumber
\BV(R_{d,k}(u)) & = & \BV(R_{d,k-1}(u) + Z_k)\\ \label{var1}
                & = & \BE\big{[} \BV(R_{d,k-1}(u) + Z_k|R_{d,k-1}(u))\big{]} \\
\nonumber       & & + \BV\left(\BE[R_{d,k-1}(u)+Z_k|R_{d,k-1}(u)]\right).
\end{eqnarray}
Since $\BV(R_{d,k-1}(u) + Z_k|R_{d,k-1}(u)) = \BV(Z_k|R_{d,k-1}(u))$
and the conditional distribution of $Z_k$ given $R_{d,k-1}(u)$ is the hypergeometric distribution
Hyp$(d,u-R_{d,k-1}(u),R_{d,k-1}(u))$,
it follows that
\begin{eqnarray}\nonumber
\BV(Z_k|R_{d,k-1}(u)) & = & d \cdot \frac{u-R_{d,k-1}(u)}{u}\\
\nonumber & & \cdot \left(1 - \frac{u-R_{d,k-1}(u)}{u} \right) \left(1- \frac{d-1}{u-1}\right)\\ \label{ersterterm}
    & = & \frac{d}{u^2} \left(1 - \frac{d-1}{u-1}\right)\\
\nonumber & & \cdot (u-R_{d,k-1}(u))R_{d,k-1}(u).
\end{eqnarray}
Moreover, we have
\begin{eqnarray*}
\BE[R_{d,k-1}(u)+Z_k|R_{d,k-1}(u)] & = & R_{d,k-1}(u) + \BE[Z_k|R_{d,k-1}(u)] \\
                                   & = & Z_{k-1} + d \cdot \frac{u-R_{d,k-1}(u)}{u}\\
                                   & = & d + \left(1-\frac{d}{u}\right) R_{d,k-1}(u).
\end{eqnarray*}
Therefore, the second summand figuring in (\ref{var1}) equals
\[
  \left(1-\frac{d}{u}\right)^2 \BV(R_{d,k-1}(u)).
\]
Since $\BV(R_{d,k-1}(u)) = \BE\left[R^2_{d,k-1}(u)\right] - (\BE R_{d,k-1}(u))^2$, (\ref{ersterterm}) yields
\begin{eqnarray*}
\nonumber \BE[\BV(Z_k|R_{d,k-1}(u))] & = & \frac{d}{u^2} \left(1 - \frac{d-1}{u-1}\right)\\
\nonumber & & \cdot \left(u\BE(R_{d,k-1}(u))\right.\\
\nonumber & & \left. \quad - \BV(R_{d,k-1}(u)) - (\BE R_{d,k-1}(u))^2\right).
\end{eqnarray*}
We thus obtain the recursion formula
\begin{eqnarray*}
\nonumber \BV(R_{d,k}(u)) & = & v\cdot \BV(R_{d,k-1}(u)) \\
\nonumber & & + \frac{d}{u^2} \cdot \left(1- \frac{d-1}{u-1}\right)\\
\nonumber & & \quad \cdot \left( u\BE(R_{d,k-1}(u)) - (\BE R_{d,k-1}(u))^2 \right)
\end{eqnarray*}
with $v$ given in (\ref{defv}),
from which the result follows by straightforward calculations.
Notice that $\BV(R_{d,1}(u))=0$ ($R_{d,1}(u)$ is the constant $d$),
and that $\lim_{k \to \infty} \BV(R_{d,k}(u)) = 0$.
The latter convergence is clear from the fact that,
in the long run, all red balls will have been drawn.\\

c) The result follows from the fact that the event $\{R_{d,k}=j\}$ is the union of the pairwise disjoint events
$\{R_{d,k-1}(u) =j-\ell, Z_k=\ell\}$, $\ell = 0,1,\ldots,d$,
and the fact that the conditional distribution of $R_{d,k}(u) (= R_{d,k-1}(u) + Z_k)$ given
$R_{d,k-1}(u) = j-\ell$ is the hypergeometric distribution Hyp$(d,u-(j-\ell),j-\ell)$.

\section{The hypergeometric distribution}
\label{sec:hyp_distribution}
Suppose an urn contains $b$ black and $w$ white balls.
If $m$ balls are drawn completely at random without replacement,
then the number $X$ of black balls drawn has the hypergeometric distribution Hyp$(m,b,w)$,
i.e., we have
\[
\PP(X= j) = \frac{{\binom{b}{j}}{\binom{w}{ m-j}}}{{\binom{b+w}{m}}}, \qquad j =0,1,\ldots,m,
\]
where we put ${\binom{s}{\ell}} := 0$ if $s < \ell$.
Expectation and variance of $X$ are given by
\begin{eqnarray*}
\BE(X) & = & m \cdot \frac{b}{b+w}, \\
\BV(X) & = & m  \cdot \frac{b}{b+w} \cdot \left( 1- \frac{b}{b+w}\right) \left( 1- \frac{m-1}{b+w-1}\right),
\end{eqnarray*}
respectively.

\section*{Acknowledgment}
We thank the Matrix developers for their ingenious system design,
and Alexander Marsteller for many hours of differential equation analysis.

\sloppy
\balance
\printbibliography

\end{document}